\documentclass{article}
\usepackage[top=2cm, bottom=2cm, right=2cm, left=2cm]{geometry}

\usepackage{amssymb, amsthm, amsmath,stmaryrd,array}
\usepackage[english]{babel}
\usepackage[numbers]{natbib}
\usepackage[utf8]{inputenc}
\usepackage{caption, subcaption}
\usepackage{graphicx,wrapfig}
\usepackage{tikz}
\usepackage[nottoc]{tocbibind}
\usepackage{environ,listings,xcolor,color}
\usepackage{hyperref}
\usepackage{colortbl}
\usepackage[shortlabels]{enumitem}
\usepackage{booktabs}
\usepackage{graphicx}
\usepackage{threeparttable}
\usepackage[capitalize]{cleveref}

\usepackage{thmtools}
\usepackage{thm-restate}

\usepackage{booktabs}                                                                                                
\usepackage{url}                                                                                                     
\usepackage{tabularx}                                                                                                
\usepackage{ragged2e}                                                                                                
\usepackage{makecell}      

\usepackage{tablefootnote}
\newcounter{step}

\renewcommand{\epsilon}{\varepsilon}
\renewcommand{\mod}{\bmod}
\newcommand{\ccite}[1]{\citeauthor{#1} \cite{#1}}

\hypersetup{
    colorlinks=true,
    linkcolor=blue,
    citecolor=blue,
    }

\allowdisplaybreaks

\newtheorem{definition}{Definition}
\newtheorem{theorem}{Theorem}
\newtheorem{lemma}[theorem]{Lemma}

\newtheorem{assumption}{Assumption}

\begin{document}


\title{Improved Bounds for Rectangular Monotone Min-Plus Product and Applications}

\author{Anita Dürr\thanks{
{anita.durr@epfl.ch (École Polytechnique Fédérale de Lausanne, Switzerland)}
}}

\date{August, 2022}

\maketitle

\begin{abstract}
    In a recent breakthrough paper, \citeauthor{chi2022faster} (STOC'22)
    introduce an $\Tilde{O}(n^{\frac{3 + \omega}{2}})$ time algorithm to compute Monotone Min-Plus Product between two square matrices of dimensions $n \times n$ and entries bounded by $O(n)$. This greatly improves upon the previous $\Tilde{O}(n^{\frac{12 + \omega}{5}})$ time algorithm and as a consequence improves bounds for its applications. Several other applications involve Monotone Min-Plus Product between rectangular matrices, and even if \citeauthor{chi2022faster}'s algorithm seems applicable for the rectangular case, the generalization is not straightforward.
    In this paper we present a generalization of the algorithm of \citeauthor{chi2022faster} to solve Monotone Min-Plus Product for rectangular matrices with polynomial bounded values.
    We next use this faster algorithm to improve running times for the following applications of Rectangular Monotone Min-Plus Product: $M$-bounded Single Source Replacement Path, Batch Range Mode, $k$-Dyck Edit Distance and 2-approximation of All Pairs Shortest Path. We also improve the running time for Unweighted Tree Edit Distance using the algorithm by \citeauthor{chi2022faster}.
\end{abstract}

\paragraph{Keywords.}
Monotone Min-Plus Product; Bounded-Difference Min-Plus Product; Fine-Grained Complexity

\section{Introduction}

The \textbf{Min-Plus Product} $C$ between two $n \times n$ matrices $A$ and $B$ is defined as the $n \times n$ matrix $C = A \star B$ with entries $C_{i, j} = \min_{k \in [n]} \{A_{i, k} + B_{k, j}\}$ for all $i, j \in [n]$. The naive algorithm computes $C$ in $O(n^3)$ time while the currently fastest algorithm by \citet{fastAPSP, fastAPSP_journal} runs in time $n^3/2^{\Theta(\sqrt{\log n})}$. However, since the All Pairs Shortest Path (APSP) problem is equivalent to the Min-Plus Product, under the APSP Hypothesis there is no algorithm solving Min-Plus Product in truly subcubic running time of $O(n^{3 - \epsilon})$ for some $\epsilon > 0$. Nevertheless, it is sometimes possible to break that barrier if we assume some structure on the matrices. In their recent breakthrough paper, \citet{chi2022faster} show that if $B$ has monotone rows then one can solve the Min-Plus Product in time $\Tilde O(n^{\frac{3 + \omega}{2}})$\footnote{We use the notation $\Tilde{O}$ to hide a polylogarithmic factor.}, where $\omega$ is the exponent of fast matrix multiplication. This improves upon the previously best known algorithm by \citet{polak} for Square Monotone Min-Plus Product which had a running time of $\Tilde O (n^{\frac{12+\omega}{5}})$.

Fortunately, Min-Plus Product, as a fundamental problem, can be used to solve many other applications in which it is enough to assume that one of the matrices is monotone. Therefore the result of \citeauthor{chi2022faster} yields improvements for several applications. A direct consequence is for instance the improvement of algorithms solving the Language Edit Distance, RNA-folding and Optimum Stack Generation problems, which where reduced to square Monotone Min-Plus Product by \citet{rna_led_osg}. However there are several other applications of Monotone Min-Plus Product in which the matrices are rectangular: the Single Source Replacement Path (SSRP) with bounded edge-weights (\ccite{polak}), Batch Range Mode (\ccite{batch_range_mode}), $k$-Dyck Edit Distance (\ccite{dyck_arxiv}), as well as a 2-approximation of APSP (\ccite{apsp_approx}).
Even though the techniques used in \cite{chi2022faster} to solve Square Monotone Min-Plus Product seem to be applicable for rectangular matrices, a generalization of the algorithm is not straightforward and needs to be done carefully. In this paper we show the generalization for the rectangular case and then improve the running time for the above mentioned problems. We also improve the running time for Unweighted Tree Edit Distance, which was reduced to Square Monotone Min-Plus Product by \citet{treeED}, using the algorithm by \citet{chi2022faster}.

\begin{table*}[h!]
    \centering
    \begin{threeparttable}
    \caption{Improved bounds of the expected running time (in $\Tilde{O}$) of randomized algorithms solving applications of Monotone Min-Plus Product. The improvements are resulting from \cref{thm:fast_rect} and are further discussed in Section \ref{sec:applications}. When possible we express the bound as a function of $\omega$, otherwise we use the current best bound on fast matrix multiplication \cite{gall_square} and fast rectangular matrix multiplication \cite{gall_rect}.}
    \label{fig:table_applications}
    \begin{tabular}{ @{}lll@{} }
     \toprule
     Problem & Previous bound & Improved bound \\
     \midrule
        $M$-bounded SSRP \cite{polak} & $M^{\frac{5}{17 - \omega}} n^{\frac{36 - 7 \omega}{17 - 4\omega}}$ or $M^{0.8043} n^{2.4957}$  & $M^{\frac{2}{5 - \omega}} n^{\frac{9 - \omega}{5 - \omega}}$ or $M^{0.8825} n^{2.4466}$ \\
        Batch Range Mode\tnote{$\dagger$}~\cite{batch_range_mode, Gao_He_batch} & $n^{\frac{18 + 2\omega}{13 + \omega}} \leq n^{1.4796}$ & $n^{\frac{6 + 2\omega}{5 + \omega}} \leq n^{1.4575}$ \\
        $k$-Dyck Edit Distance \cite{dyck_arxiv} & $n + k^{4.7820}$ & $n + k^{\frac{5}{2} + \omega(1/2)} \leq n + k^{4.5442}$ \\
        $2$-approximation APSP \cite{apsp_approx} & $n^{2.2867}$ & $n^{2.2593}$ \\
        Unweighted Tree Edit Distance \cite{treeED} & $n\cdot m^{1.9541}$ & $n \cdot m^\frac{3\omega + 7}{\omega + 5} \leq n \cdot m^{1.9149}$ \\
     \bottomrule
    \end{tabular}
    \footnotesize
    \begin{tablenotes}
    \item[$\dagger$] The first version of this paper erroneously claimed a lower improved bound.
    \end{tablenotes}
    \end{threeparttable}
\end{table*}

\paragraph{Definitions}
We restrict ourselves to integer matrices only. A matrix $X$ is said to be \textbf{monotone} if each row is a non-decreasing sequence and every entry is non-negative and bounded by $O(n^\mu)$ for some $\mu \geq 0$.
When the monotonicity holds for the columns we call the matrix column-monotone.
We also consider Min-Plus Product between rectangular matrices $A$ and $B$ of dimensions $n \times n^\beta$ and $n^\beta \times n$ for some $\beta \geq 0$. Define $\omega(\beta)$ to be the \textbf{exponent of fast rectangular matrix multiplication} between two matrices of dimensions $n \times n^\beta$ and $n^\beta \times n$. So $\omega = \omega(1)$. We focus on the \textbf{Rectangular Monotone Min-Plus Product} parameterized by $(\beta, \mu)$, which is the Min-Plus Product between matrices $A$ and $B$ of dimensions $n \times n^\beta$ and $n^\beta \times n$ where $B$ is monotone with non-negative entries bounded by $O(n^\mu)$.

\paragraph{Results}
Our contribution is twofold. First we provide
an algorithm computing the Rectangular Monotone Min-Plus Product in time $\Tilde O(n^{\frac{1 + \beta + \mu + \omega(\beta)}{2}})$, as stated in \cref{thm:fast_rect}. We remark that the algorithm directly generalizes the algorithm provided by \citet{chi2022faster} for the special setting $\beta = \mu = 1$ and that in that case one falls back on their results. Next, we detail how the result of \cref{thm:fast_rect} is used to improve running times of $M$-bounded SSRP, $k$-Dyck Edit Distance, 2-approximation of APSP. Indeed, in those applications the dependency on the Monotone Min-Plus Product is not direct but depends on some parameters that need to be optimized according to the new running time given by \cref{thm:fast_rect}. 
Moreover, we explain how a small modification to the Batch Range Mode algorithm of \citet{batch_range_mode} allows using the (square) Monotone Min-Plus Product of \citet{chi2022faster} to obtain an improved running time.
We also simplify and improve the running time for Unweighted Tree Edit Distance. All new bounds are summarized in \cref{fig:table_applications}. Additionally, we discuss one consequence of \citeauthor{chi2022faster}'s algorithm for the lower bound of SSRP algorithms.

\begin{restatable}{theorem}{genalgo}
    \label{thm:fast_rect}
    Let $\beta, \mu$ be non-negative real numbers. Let $A$ be an $n \times n^\beta$ integer matrix and $B$ an $n^\beta \times n$ integer monotone matrix with non-negative entries bounded by $O(n^\mu)$. Then the Min-Plus Product $A \star B$ can be computed in $\Tilde{O}(n^{\frac{1 + \beta + \mu + \omega(\beta)}{2}})$ time in expectation.
\end{restatable}

A similar result for \cref{thm:fast_rect} and the improved running time for the $k$-Dyck Edit Distance problem was obtained independently and in parallel by \ccite{dyck_soda}.

\paragraph{Monotone and Bounded-Difference Min-Plus Product}
Strict\-ly speaking, except for SSRP and Batch Range Mode, the above mentioned applications don't use Monotone Min-Plus Product but Bounded-Difference Min-Plus Product. A matrix $X$ is \textbf{$\delta$-bounded-difference} (or simply bounded-difference if $\delta$ is a constant) if $|X_{i, j} - X_{i, j+1}| \leq \delta$ and $|X_{i, j} - X_{i + 1, j}| \leq \delta$ for all $i, j$. One can reduce a $\delta$-bounded difference matrix $X$ to a monotone matrix $X'$ in quadratic time by setting $X'_{i, j} = X_{i, j} + \delta j - X_{1, 1}$ for all $i, j$. In that case elements of $X'$ are non-negative and bounded by $O(n\delta)$.
In most of the above mentioned applications, a Min-Plus Product is performed between two bounded-difference matrices. This means that a stronger structure is assumed on the matrices than just the monotonicity of one matrix and this could allow for faster algorithms. In fact, before the recent result of \citet{chi2022faster} (and its generalization via \cref{thm:fast_rect}), the best known algorithm by \citet{chi_bdd_diff} for Square Bounded-Difference Min-Plus Product had a running time of $\Tilde O(n^{2.779})$, while the best known algorithm by \citet{polak} for square Monotone Min-Plus Product (when $\mu=1$) had a running time of $\Tilde O(n^{2.8653})$. However \cref{thm:fast_rect} improves on both of those bounds and thus improves the running time for both Monotone and Bounded-Difference Min-Plus Product.

\paragraph{Related work}
It is worth mentioning that other problems such as Dynamic Range Mode (\ccite{polak}), 3SUM+ and 3SUM+ in preprocessed universe (\ccite{lewenstein_chan}) were solved alongside Min-Plus Product. \citeauthor{lewenstein_chan} reduce Min-Plus Convolution to 3SUM+ and a better algorithm for Min-Plus Convolution thus doesn't entail a better algorithm for 3SUM+. However the tools and techniques used in \ccite{lewenstein_chan} for Min-Plus Convolution are adapted to 3SUM+ and 3SUM+ in a preprocessed universe to obtain a faster algorithm. Similarly, in \ccite{polak} the Min-Plus Product algorithm is adapted for Dynamic Range Mode. In both those work Min-Plus Product is not used as a black-box, hence our result does not directly imply an improvement.

Finally, we would like to mention that \citeauthor{chi2022faster} also solve Monotone Min-Plus Convolution in $\Tilde O(n^{1.5})$ time with the same technique as in \cite{chi2022faster} for Monotone Min-Plus Product. This is a major improvement upon the previously $\Tilde{O}(n^{1.869})$ time algorithm by \citet{lewenstein_chan}. Since Binary Jumbled Indexing (or Histogram Indexing) can be reduced to Monotone Min-Plus Convolution in quadratic time (see \citet{lewenstein_chan}), this also directly yields an $\Tilde O(n^{1.5})$ time algorithm for it. Additionally, earlier this year \citet{knapsack} used \citeauthor{chi2022faster}'s result to get new running times for Knapsack.

\section{Rectangular Monotone Min Plus Product}
\label{sec:algo}
In this section we reformulate and generalize the algorithm presented by \citet{chi2022faster} to compute the Monotone Min-Plus Product for arbitrary $\beta, \mu \geq 0$, thus proving \cref{thm:fast_rect}. We recall that all the ideas of the presented algorithm are from \citet{chi2022faster} and that the only novelty is to extend the setting of the matrices.
We also like to note that the following algorithm can be adapted to solve Min-Plus Product when $B$ is column-monotone instead of row-monotone by using the same method as explained by \citeauthor{chi2022faster} in \cite[Appendix A]{chi2022faster}. The running time is slightly different and we sketch a proof in \cref{sec:BMMP_col}. Also, since $(A \star B)^T = B^T \star A^T$, this algorithm  and the one described in \cref{sec:BMMP_col} can be used if $A$ is column- or row-monotone respectively.

The algorithm uses the \textbf{segment tree} data structure which allows for a sequence of integers $S = \{s_1, \dots s_n\}$ to perform the following update and query operations in $O(\log n)$ deterministic worst-case time. An update on $S$ is defined by an interval $[i, j] \subset [n]$ and an integer $u$ and updates $s_l \gets \min\{s_l, u\}$ for each $l \in [i, j]$. A query operation returns the current value of an arbitrary element $s_l$.
We will also \textbf{multiply matrices of polynomials}: if $A$ and $B$ are matrices of polynomials of degree at most $d$ and of dimensions $n \times n^\beta$ and $n^\beta \times n$, then we can compute the product $A \cdot B$ in $\Tilde O(d n^{\omega(\beta)})$ time using Fast Fourier Transform and fast matrix multiplication.

\paragraph{Overview}
Let $A$ and $B$ be two matrices of dimensions $n \times n^\beta$ and $n^\beta \times n$ such that $B$ is monotone and with entries bounded by $O(n^\mu)$. The idea of the algorithm is to divide the matrix $C = A \star B$ into two matrices $C^*$ and $\Tilde{C}$ such that $C^*_{i, j} = \lfloor C_{i, j} / p \rfloor$ and $\Tilde{C}_{i, j} = C_{i, j} \mod p$ for some integer $p$ that we describe later. This way, the original value $C_{i, j}$ can be computed by $p \cdot C^*_{i, j} + \Tilde{C}_{i, j}$.
To efficiently compute the quotient matrix $C^* = \lfloor C /p \rfloor$, we will construct matrices $A^* = \lfloor A /p \rfloor$ and $B^* = \lfloor B /p \rfloor$ and use a combinatorial approach to determine the entries of $C^*$.

The computation of the remainder matrix $\Tilde{C} = (C \mod p)$ is more complex. We will proceed iteratively, each step computing one bit of $\Tilde{C}$. For that, we define $h$ to be such that $2^{h-1} \leq p < 2^h$ and for each $l = h, h-1, \dots, 0$ let $A^{(l)} = \left\lfloor \frac{A \mod p}{2^l} \right\rfloor$ and $B^{(l)} = \left\lfloor \frac{B \mod p}{2^l} \right\rfloor$. We then construct a matrix $C^{(l)}$ approximating $\left\lfloor \frac{C \mod p}{2^l} \right\rfloor$ such that in the end $C^{(0)} = \Tilde{C}$. Note that $C^{(l)}$ will not be the Min-Plus Product $A^{(l)} \star B^{(l)}$, as this would be too costly, but is only an approximation of it. $C^{(l)}$ is computed using $C^{(l+1)}, A^{(l + 1)}, B^{(l + 1)}, A^{(l)}$ and $B^{(l)}$. \cref{lemma:filter} ensures that the correct $C^{(l)}_{i, j}$ is equal to the sum $A^{(l)}_{i, k} + B^{(l)}_{k, j}$ for some $k$ satisfying both conditions:
\begin{itemize}
    \item $A^{(l + 1)}_{i, k} + B^{(l + 1)}_{k, j} = C^{(l+1)}_{i, j} + b \text{ for some } b \in \{-10, \dots, 10\}$
    \item $A_{i, k}^{*} + B^*_{k, j} = C^*_{i, j}$
\end{itemize}
The algorithm first focuses on the first condition, and then removes all candidates that don't satisfy the second condition. In fact,  terms  $A^{(l)}_{i, k} + B^{(l)}_{k, j}$ satisfying the first conditio can be filtered out by using a commun polynomial multiplication trick: construct two matrices of polynomials $A^p, B^p$ on two variables $x$ and $y$ such that for any $i \in [n], k \in [n^\beta], j \in [n]$
$$A^p_{i, k} = x^{A_{i, k}^{(l)}} \cdot y^{A_{i, k}^{(l + 1)}}$$
and
$$B^p_{k, l} = x^{B_{k, j}^{(l)}} \cdot y^{B_{k, j}^{(l + 1)}}.$$
When computing the (standard) product $C^p = A^p \cdot B^p$, we will get for any $i, j \in [n]^2$ $$
C^p_{i, j} = \sum_{k \in [n^\beta]} x^{A_{i, k}^{(l)} + B_{k, j}^{(l)}} \cdot y^{A_{i, k}^{(l + 1)} + B_{k, j}^{(l + 1)}}.
$$
So to filter out terms $A_{i, k}^{(l)} + B_{k, j}^{(l)}$ satisfying the first condition of \cref{lemma:filter}, we can compare the $y$-degrees of $C^p_{i, j}$ to $C^{(l+1)}_{i, j} + b$ for every offset $b \in \{-10, \dots, 10\}$ and keep the corresponding $x$-degrees. Now it might be the case that those terms don't satisfy the second condition, i.e.\ ~are such that $A_{i, k}^{*} + B^*_{k, j} \neq C^*_{i, j}$. This is why we construct for all $b \in \{-10, \dots, 10\}$ and $l \in \{h, \dots, 0\}$ the set $T_b^{(l)}$ of all terms such that the first condition of Lemma \ref{lemma:filter} holds, but not the second. We then subtract the erroneous terms selected by $T_b^{(l)}$ to keep only the terms of Lemma \ref{lemma:filter}. In order to efficiently compute the $T_b^{(l)}$ we will use previously computed $T_b^{(l+1)}$. To bound the size of the set  $T_b^{(l)}$ we will use the monotonicity of $B^{(l)}, C^{(l)}$ but we will also need the integer $p$ to be a prime number (see \cref{lemma:bound_Tb}).

Finally, a last detail to address is the running time for computing $C^p$.
The $x$ and $y$-degrees are bounded by
$O(n^{\alpha})$ and play a multiplicative factor in the computation of $C^p$.
We can improve this by replacing the $x$-degrees by $A_{i, k}^{(l)} - 2A_{i, k}^{(l+1)}$ and $B_{k, j}^{(l)} - 2B_{k, j}^{(l+1)}$ respectively. Then the $x$-degree becomes 0 or 1, and the computation of $C^p$ takes $\Tilde O(n^{\alpha} \cdot n^{\omega(\beta)})$ time.
\newline




\subsection{Algorithm}
We now describe the full algorithm in detail. Let $A$ and $B$ be matrices as in Theorem \ref{thm:fast_rect} of dimensions $n \times n^\beta$ and $n^\beta \times n$ and such that $B$ is monotone with $O(n^\mu)$ bounded entries for some reals $\beta, \mu \geq 0$. Following the same reasoning as in \cite{chi2022faster}, we can assume that all entries of $A$ are non-negative and at most $O(n^\mu)$. As a consequence, the matrix $C = A \star B$ is row-monotone with non-negative entries at most $O(n^\mu)$.

Without loss of generality, we can assume that $n$ is a power of 2. Define a constant parameter $\alpha \in (0, \mu)$ to be optimized later. Pick uniformly at random a prime number $p$ in the range $[40n^\alpha, 80n^\alpha]$ and define the integer $h$ such that $2^{h-1} \leq p < 2^h$. Similarly as in \cite{chi2022faster}, we can make the following assumption (see \cref{lemma:ass}).

\begin{assumption}\label{ass:p/3}
For every $i, j$ either $(A_{i, j}~\mod~p)~<~p/3$ or $A_{i, j}~=~+\infty$. For every $i, j$ $(B_{i, j}~\mod~p)~<~p/3$ and each row of $B$ is non-decreasing.
\end{assumption}

The idea of the algorithm is to compute two matrices $C^*$ and $\Tilde{C}$ such that $C^*_{i, j} = \lfloor C_{i, j} / p \rfloor$ and $\Tilde{C}_{i, j} = C_{i, j} \mod p$ for each $i, j$. Then, to recover the original value $C_{i, j}$ simply compute $p \cdot C^*_{i, j} + \Tilde{C}_{i, j}$. We compute $C^*$ using a combinatorial approach, while we compute $\Tilde C$ iteratively, computing each bit at a time.

\paragraph{Compute the quotient matrix $\lfloor C /p \rfloor$} Define $A^*$ and $B^*$ as $A^*_{i, j} = \lfloor A_{i, j} / p \rfloor$ if $A_{i, j}$ is finite, otherwise $A^*_{i, j} = +\infty$, and $B^*_{i, j} = \lfloor B_{i, j} / p \rfloor$. Then, by Assumption \ref{ass:p/3}, the product $C^* = A^* \star B^*$ is such that if $C_{i, j}$ is finite then $C^*_{i, j} = \lfloor C_{i, j} / p \rfloor$.

Note that $B^*$ is row-monotone with entries bounded by $O(n^{\mu - \alpha})$. Thus in each row $k$ of $B^*$ we can define at most $O(n^{\mu - \alpha})$ intervals $[j_0, j_1] \subset [n]$ such that $B^*_{k, j_0 -1} \neq B^*_{k, j_0} = B^*_{k, j_0 + 1} = \dots = B^*_{k, j_1} \neq B^*_{k, j_1 + 1}$. To compute $C^*$ we first initialize each entry to $+\infty$. We then loop over each $i \in [n]$, each $k \in [n^\beta]$ and each interval $[j_0, j_1] \subset [n]$ as defined above for the $k$-th row of $B^*$ and update $C^*_{i, j} \gets \min \{C^*_{i, j}, A^*_{i, k} + B^*_{k, j_0}\}$ for each $j \in [j_0, j_1]$. This operation is correct since $B^*_{k, j} = B^*_{k, j_0}$ for every $j \in [j_0, j_1]$. Updating all elements $C^*_{i, j}$ for $j \in [j_0, j_1]$ can be done in logarithmic time using a segment tree. Hence the total running time to compute $C^*$ is the range of $i$ times the range of $k$ times the number of intervals in each row of $B^*$, i.e.\ $\Tilde{O}(n^{1 + \beta  + \mu - \alpha})$.

\paragraph{Compute the remainder matrix $(C \mod p)$} For every $0 \leq l \leq h$, construct $n \times n^\beta$ and $n^\beta \times n$ matrices $A^{(l)}$ and $B^{(l)}$ as $A^{(l)}_{i, j} = \left\lfloor \frac{A_{i, j} \mod p}{2^l} \right\rfloor$ if $A_{i, j}$ is finite, $A^{(l)}_{i, j} = +\infty$ otherwise and  $B^{(l)}_{i, j} = \left\lfloor \frac{B_{i, j} \mod p}{2^l} \right\rfloor$. We will iteratively construct an $n \times n$ matrix $C^{(l)}$ approximating $\left\lfloor \frac{C_{i, j} \mod p}{2^l} \right\rfloor$ if $C_{i, j}$ is finite. $C^{(l)}$ isn't defined directly but constructed such that it satisfies the following properties if $C_{i, j}$ is finite:
\begin{enumerate}[(1)]
    \item $\left\lfloor \frac{C_{i, j} \mod p \ - \ 2(2^l -1)}{2^l} \right\rfloor \leq C_{i, j}^{(l)} \leq \left\lfloor \frac{C_{i, j} \mod p \ + \ 2(2^l -1)}{2^l}  \right\rfloor$
    \item for $j_0 <j_1$, if $C^*_{i, j_0} = C^*_{i, j_1}$ then $C^{(l)}_{{i, j_0}}, \dots, C^{(l)}_{{i, j_1}}$ are monotonically non-decreasing
\end{enumerate}
Note that $C^{(l)}$ is not necessary equal to $A^{(l)} \star B^{(l)}$. We iteratively calculate matrix $C^{(l)}$ for $l = h, h-1, \dots, 0$ satisfying the above properties. Since $\left\lfloor \frac{C_{i, j} \mod p}{2^l} \right\rfloor$ can be seen as the approximation of $C_{i, j} \mod p$ up to the $l$-th bit, at each iteration we refine our approximation and  get in the end $\Tilde{C} = C^{(0)}$. To compute $C^{(l)}$ from $C^{(l+1)}$ we look for the minimum value of $A_{i, k}^{(l)} + B_{k, j}^{(l)}$ such that $k$ satisfies both conditions:
\begin{itemize}
    \item $A^{(l + 1)}_{i, k} + B^{(l + 1)}_{k, j} = C^{(l+1)}_{i, j} \ +\ b$ for some $b \in \{-10, \dots, 10\}$
    \item $A_{i, k}^{*} + B^*_{k, j} = C^*_{i, j}$
\end{itemize}
By \cref{lemma:filter} this ensures that $C^{(l)}$ satisfies both properties (1) and (2).

We will filter out terms according to the first condition and then remove the ones that don't satisfy the second condition. We explain later how the filtering according to the first condition is done and for now define the set of \textit{erroneous terms}, which are the terms satisfying the first condition but not the second.
For that, let us first make a series of observations.
All elements in $A^{(l)}, B^{(l)}, C^{(l)}$ are infinite or non-negative integers at most $O(n^{\alpha}/2^l)$. Furthermore, by property (2) of $C^{(l)}$, and since $B$ is row-monotone and elements of $B$ and $C$ are bounded by $O(n^\mu)$, every row of $B^{(l)}$ and $C^{(l)}$ is composed of $O(n^\mu/2^l)$ intervals of equal elements. Also, each pair of rows of $B^{(l)}, C^{(l)}$ can be divided into $O(n^\mu/2^l)$ segments, where we define a segment as follow:
\begin{definition}
For $i, j_0, j_1 \in [n], k \in [n^\beta]$ and $j_0 \leq j_1$ define a segment w.r.t $B^{(l)}$ and $C^{(l)}$ to be the tuple $(i, k, [j_0, j_1])$ such that for all $j \in [j_0, j_1]$ $B^{(l)}_{k, j} = B^{(l)}_{k, j_0}$, $B^{*}_{k, j} = B^{*}_{k, j_0}$, and $C^{(l)}_{i, j} = C^{(l)}_{i, j_0}$, $C^{*}_{i, j} = C^{*}_{i, j_0}$.
\end{definition}

\noindent We can then define the set of erroneous terms $T_b^{(l)}$ for each offset $b \in \{-10, \dots, 10\}$ and every $l \in \{0, \dots, h\}$ as the set of segments $(i, k, [j_0, j_1])$
w.r.t  $B^{(l)}, C^{(l)}$ such that
$A_{i, k} < +\infty,\ A^{(l)}_{i, k} + B^{(l)}_{k, j_0} = C^{(l)}_{i, j_0} + b$ and
$A^{*}_{i, k} + B^{*}_{k, j_0} \neq C^{*}_{i, j_0}$.
We will show in Lemma \ref{lemma:bound_Tb} that the size of $T_b^{(l)}$ can be bounded by $\Tilde{O}(n^{1 + \beta + \mu -\alpha})$.

\paragraph{First iteration} The algorithm starts with $l = h$, so $A^{(l)}$, $B^{(l)}$, $C^{(l)}$ are zero matrices in the first iteration. $T_b^{(h)}$ is an empty set for $b \neq 0$ and $T_0^{(h)}$ is the set of segments $(i, k, [j_0, j_1])$ (w.r.t $B^{(h)}, C^{(h)}$) such that $A_{i,k} < +\infty$ and $A^{*}_{i, k} + B^{*}_{k, j_0} \neq C^{*}_{i, j_0}$. Since the number of segments in each pair of rows of $B^{(h)}, C^{(h)}$ is at most $O(n^\mu/p) = O(n^{\mu - \alpha})$ and since the number of pairs of rows of $B^{(h)}, C^{(h)}$ is $O(n^{1 + \beta})$, we can bound the size of $T_{b}^{(h)}$ for every $b \in \{-10, \dots, 10\}$ by $O(n^{1 + \beta + \mu - \alpha})$. Hence the first iteration does not take more than $O(n^{1 + \beta + \mu - \alpha})$ time.
\newline

Each following iteration of $l = h-1, \dots, 0$ consists of three phases: a polynomial matrix multiplication using $A^{(l)}$, $A^{(l+1)}$, $B^{(l)}$ and $B^{(l+1)}$ that allows to filter out candidates for $C^{(l)}$; a subtraction of the erroneous terms $T_b^{(l+1)}$; and finally the computation of $T_b^{(l)}$ using $T_b^{(l+1)}$. We describe the three phases for a fixed $l \in \{0, \dots, h-1\}$.

\paragraph{Multiply matrices of polynomials}
Construct two polynomial matrices $A^p$ and $B^p$ on variables $x, y$ as $A^p_{i, k} = 0$ if $A_{i, k}$ is infinite, otherwise:
$$A^p_{i, k} = x^{A^{(l)}_{i, k} - 2A^{(l + 1)}_{i, k}} \cdot y^{A^{(l+1)}_{i, k}}$$ and $$B^p_{k, j} = x^{B^{(l)}_{k, j} - 2B^{(l + 1)}_{k, j}} \cdot y^{B^{(l+1)}_{k, j}}.$$
Compute the standard matrix multiplication $C^p = A^p \cdot B^p$. Since $A^{(l)}_{i, k} - 2A^{(l + 1)}_{i, k}$ and $B^{(l)}_{k, j} - 2B^{(l + 1)}_{k, j}$ are 0 or 1, the $x$-degree is bounded by 1. The $y$-degree is bounded by $O(p)= O(n^\alpha)$. Hence computing $C^p$ takes $\Tilde{O}(n^{\omega(\beta) + \alpha})$ time.

\paragraph{Filter candidates and subtract erroneous terms}
The previous computation of $C^p$ allows us to filter out candidate terms $A^{(l)}_{i, k}$ and $B^{(l)}_{k, j}$ that satisfy the first condition for $C^{(l)}_{i, j}$, i.e.\ $A^{(l + 1)}_{i, k} + B^{(l + 1)}_{k, j} = C^{(l+1)}_{i, j} + b$ for some $b \in \{-10, \dots, 10\}$.
Indeed we have that for every $i, j \in [n]^2$:
$$C^p_{i, j} = \sum_{\substack{k \in [n^\beta] \\ A_{i,k} < + \infty}} x^{A^{(l)}_{i, k} + B^{(l)}_{k, j}- 2A^{(l + 1)}_{i, k} - 2B^{(l + 1)}_{k, j}} \cdot y^{A_{i, k}^{(l+1)} + B_{k, j}^{(l+1)}}.$$
If $C^p_{i, j} = 0$, set $C^{(l)}_{i, j} = +\infty$. Otherwise, build the following polynomial for each $b \in \{ -10, \dots, 10\}$:
$$
C^p_{i, j, b} (x) = \sum_{\substack{\lambda x^c \cdot y^d \text{ is a term in } C^p_{i, j} \\ d = C^{(l+1)}_{i, j} + b}} \lambda x^c.
$$
This construction takes $\Tilde{O}(n^{1+\beta+\alpha})$ time for all $b \in \{-10 ,\dots, 10\}$ and $i, j \in [n]^2$.
We then compute the polynomial corresponding to the erroneous terms:
$$
R^p_{i, j, b} (x) = \sum_{\substack{j \in [j_0, j_1] \\ (i, k, [j_0, j_1]) \in T_b^{(l+1)}}} x^{A^{(l)}_{i, k} + B^{(l)}_{k, j}- 2A^{(l + 1)}_{i, k} - 2B^{(l + 1)}_{k, j}}
.$$
Since $B^{(l)}_{k, j} - 2B^{(l + 1)}_{k, j}$ is 0 or 1, there can be at most two different segments $(i, k, [j_0, j_1]) \in T_b^{(l+1)}$ corresponding to $B^{(l)}_{k, j}$. This means that for a fixed $b \in \{-10, \dots, 10\}$, we can construct the polynomials $R^p_{i, j, b} (x)$ for all $i, j \in [n]^2$ by considering each segment $(i, k, [j_0, j_1])$ of $T_b^{(l+1)}$ and adding the corresponding term to at most two different $R^p_{i, j, b} (x)$ and $R^p_{i', j', b} (x)$. This takes $\Tilde{O}(|T_b^{(l+1)}|)$ time for each offset $b$.


Finally, set $C^{(l)}_{i, j} = \min_{b \in \{-10, \dots, 10\}} \{s_{i, j, b} + 2(C^{(l+1)}_{i, j} + b)\}$, where $s_{i, j, b}$ denotes the smallest degree of $x$ monomials in the difference $C^p_{i, j, b} (x) - R^p_{i, j, b} (x)$. This phase takes $\Tilde{O}(n^{1 + \beta + \alpha} + |T_b^{(l + 1)}|)$ time in total, which by Lemma \ref{lemma:bound_Tb} is $\Tilde{O}(n^{1 + \beta + \mu - \alpha})$.

\subparagraph{Computing $T_{b}^{(l)}$ from $T_{b}^{(l + 1)}$.}
By Lemma \ref{lemma:bound_Cl}, both $B^{(l)}_{k, j} - 2B^{(l + 1)}_{k, j}$ and $C^{(l)}_{i, j} - 2C^{(l + 1)}_{i, j}$ are between 0 and $O(1)$. So each segment w.r.t.\ $B^{(l+1)}$, $C^{(l+1)}$ contains at most~$O(1)$ segments w.r.t.\ $B^{(l)}$, $C^{(l)}$. Furthermore, by \cref{lemma:bigcup}, we have $\bigcup\limits_{b=-10}^{10} T_b^{(l)} \subset \bigcup\limits_{b=-10}^{10} T_b^{(l + 1)}$. So to construct $T_{b}^{(l)}$, we need to look at each sub-segment of each segment in $\bigcup\limits_{b=-10}^{10} T_b^{(l + 1)}$ and put it in the set $T_{b}^{(l)}$ it belongs to. Each segment in $T_b^{(l + 1)}$ can be split in at most $O(1)$ sub-segments and we can use binary search to find the splitting points. Hence the construction of $T_{b}^{(l)}$ for all $b \in \{-10, \dots, 10\}$, takes $\Tilde{O}(|\bigcup\limits_{b=-10}^{10} T_b^{(l + 1)} | )$ time, which by Lemma \ref{lemma:bound_Tb} is $\Tilde{O}(n^{1 + \beta + \mu - \alpha})$.

\paragraph{Total running time} The expected running time of the algorithm is bounded by $\Tilde{O}(n^{1 + \beta + \mu - \alpha} + n^{\omega(\beta) + \alpha})$. Setting $\alpha = \frac{1 + \beta + \mu - \omega(\beta)}{2}$ we equalize both terms and get $\Tilde{O}(n^{\frac{1 + \beta + \mu + \omega(\beta)}{2}})$.

\subsection{Proof of correctness}\label{sec:algo_corr}

\begin{lemma}\label{lemma:ass}
The computation of $A\star B$ where $A$ and $B$ are from \cref{thm:fast_rect} can be done by computing a constant number of Min-Plus Products $A^i \star B^i$ of matrices  $A^i, B^i$ satisfying Assumption \ref{ass:p/3}.
\end{lemma}

The proof of analogous Lemma 3.4 in \cite{chi2022faster} applies verbatim. The idea is to construct three copies of A and three copies of B containing the values respectively in $[0, p/3)$, $[p/3, 2p/3)$ and $[2p/3, p)$ when taking the modulo $p$. We can then remove an appropriate offset to get values modulo $p$ at most $p/3$, while maintaining monotonicity of rows in the copies of $B$.

\begin{lemma}\label{lemma:bound_Tb}
In expectation, for every $b \in \{-10, \dots, 10\}$ and every $l \in \{0, \dots, h\}$, $T_b^{(l)}$ contains $\Tilde{O}(n^{1 + \beta + \mu - \alpha})$ segments.
\end{lemma}

\begin{proof}
Assume that $2^l < p /100$ and consider a segment $(i, k, [j_0, j_1])$ w.r.t $B^{(l)}, C^{(l)}$. Take $j \in [j_0, j_1]$ where $A_{i, k}$ is finite and $A_{i, k}^* + B^*_{k, j} \neq C^*_{i, j}$. By Assumption \ref{ass:p/3} we have that $(C_{i, j} \mod p) < 2p /3$, hence $|A_{i, k} + B_{k, j} - C_{i, j}| \geq p/3$. Indeed, if $A^*_{i, k} + B^*_{k, j} \geq C_{i, j}^* + 1$ then

\begin{align*}
\frac{A_{i, k}}{p} + \frac{B_{k, j}}{p} &\geq \left\lfloor \frac{A_{i, k}}{p}\right\rfloor + \left\lfloor\frac{B_{k, j}}{p}\right\rfloor \geq \left\lfloor\frac{C_{i, j}}{p}\right\rfloor + 1 \geq \frac{C_{i, j}}{p} - \frac{2}{3} + 1 = \frac{C_{i, j}}{p} + \frac{1}{3},
\end{align*}
so $A_{i, k} + B_{k, j} \geq C_{i, j} + p /3$. Similarly if $A^*_{i, k} + B^*_{k, j} \leq C_{i, j}^* - 1$ then
\begin{align*}
\frac{A_{i, k}}{p} - \frac{1}{3} + \frac{B_{k, j}}{p} - \frac{1}{3} &\leq \left\lfloor \frac{A_{i, k}}{p}\right\rfloor + \left\lfloor\frac{B_{k, j}}{p}\right\rfloor \leq \left\lfloor\frac{C_{i, j}}{p}\right\rfloor - 1 \leq \frac{C_{i, j}}{p} - 1.
\end{align*}
So we have $|A_{i, k} + B_{k, j} - C_{i, j}| \geq p/3$ in both cases.

We want to bound the probability that $(i, k, [j_0, j_1])$ appears in $T_b^{(l)}$, i.e.\ the probability that
$
\left\lfloor \frac{A_{i, k} \mod p}{2^l} \right\rfloor + \left\lfloor \frac{B_{k, j} \mod p}{2^l} \right\rfloor = C^{(l)}_{i, j} + b
$
holds, which is
$$
-4 \leq \frac{A_{i, k} \mod p}{2^l} + \frac{B_{k, j} \mod p}{2^l} - \frac{C_{i, j} \mod p}{2^l} - b \leq 4.
$$
Let $C_{i, j} = A_{i, q} + B_{q, j}$. So if $(i, k, [j_0, j_1]) \in T_b^{(l)}$ then $(A_{i, k} + B_{k, j} - A_{i, q} - B_{q, j}) \mod p \in [2^l (b - 4), 2^l(b + 4)]$. This means that $A_{i, k} + B_{k, j} - A_{i, q} - B_{q, j}$ is congruent modulo $p$ to one of the $O(2^l)$ remainders. It holds that for all remainder $r \in [2^l(b -1, 2^l(b+4))]$, $(|b| \leq 10)$,
$$
|r| \leq 14 \cdot 2^l < p/6 \leq \frac{1}{2}|A_{i, k} + B_{k, j} - A_{i, q} - B_{q, j}|.
$$
If $A_{i, k}, A_{i, q}$ are finite and $B_{k, j}, B_{q, j}$ are from the original $B$ (see Lemma \ref{lemma:ass}), $|A_{i, k} + B_{k, j} - A_{i, q} - B_{q, j} -r|$ is a positive number bounded by $O(n^\mu)$. The probability $P$ that $p \in [40n^\alpha, 80n^\alpha]$ divides $A_{i, k} + B_{k, j} - A_{i, q} - B_{q, j} -r$ is the quotient of the number of divisors of
$A_{i, k} + B_{k, j} - A_{i, q} - B_{q, j} -r$ over the number of primes in $[40n^\alpha, 80n^\alpha]$. There are at most $O(\log n^\mu)$ divisors of a $O(n^\mu)$ number. By the prime number Theorem (\ccite{jameson2003prime}), there is at least $\Omega(n^\alpha/ \log n^\alpha)$ prime numbers in $[40n^\alpha, 80n^\alpha]$.  Hence $P = O(\frac{\log n^\mu \log n^\alpha}{n^\alpha}) = \Tilde{O}(n^{-\alpha})$.

In Lemma \ref{lemma:ass}, if $B_{k, j}, B_{q, j}$ are not from the original $B$ then they are set artificially to numbers congruent modulo $p$ to 0, $\lceil p/3 \rceil$ or $\lceil 2p/3 \rceil$. In that case, it still holds that $p$ divides $A_{i, k} + B_{k, j} - A_{i, q} - B_{q, j} -r$ with probability $\Tilde{O}(n^{-\alpha})$. To see this we can condition on different cases.
For instance, if $B_{k, j}$ is congruent modulo $p$ to $\lceil p/3 \rceil$ and $B_{q, j}$ is from the original $B$, then we want that $p$ divides $A_{i, k} + \lceil p/3 \rceil - A_{i, q} - B_{q, j} -r$. Since 3 does not divide $p$, $3 \lceil p/3 \rceil$ is $p+1$ or $p+2$. So $p$ divides $3(A_{i, k}  - A_{i, q} - B_{q, j} -r) + 1$ or  $3(A_{i, k}  - A_{i, q} - B_{q, j} -r) + 2$. The probability that this is the case is $\Tilde{O}(n^{-\alpha})$ in both cases. Other cases of $B_{k, j}, B_{q, j}$ are similar. Summing conditional probabilities for all cases, we obtain that $p$ divides $A_{i, k} + B_{k, j} - A_{i, q} - B_{q, j} -r$ with probability $\Tilde{O}(n^{-\alpha})$.

Finally, there are $O(2^l)$ such possible remainders $r$, so the probability that $(i, k, [j_0, j_1])$ appears in $T_b^{(l)}$ is $\Tilde{O}(2^l n^{-\alpha})$. By linearity of expectation, since there are $O(\frac{n^{1 + \beta + \mu}}{2^l})$ segments $(i, k, [j_0, j_1])$ w.r.t.\ $B^{(l)}, C^{(l)}$, the expected number of segments in $T_{b}^{(l)}$ is $\Tilde{O}(n^{1 + \beta  + \mu - \alpha})$.
\end{proof}

\begin{lemma}\label{lemma:bound_Cl}
In each iteration $l = h, \dots, 0$ and for every $i, j \in [n]^2$, we have that $C_{i, j}^{(l)} - 2C_{i, j}^{(l+1)} \in [-7, 8]$.
\end{lemma}

\begin{proof}
The proof of analogous Lemma 3.5 in \cite{chi2022faster} applies verbatim.
\end{proof}

\begin{lemma}\label{lemma:bigcup}
For every $l=h-1, \dots, 0$, we have that $\bigcup\limits_{b=-10}^{10} T_b^{(l)} \subset \bigcup\limits_{b=-10}^{10} T_b^{(l + 1)}$.
\end{lemma}

\begin{proof}
The proof of analogous Lemma 3.6 in \cite{chi2022faster} applies verbatim.
\end{proof}

\begin{lemma}
For every $l=h, \dots, 0$, if $A_{i, k} + B_{k, j} = C_{i, j}$, then there is some $b \in \{-10, \dots, 10\}$ such that $A_{i, k}^{(l)} + B_{k, j}^{(l)} = C_{i, j}^{(l)} + b$.
\end{lemma}

\begin{proof}
The proof of analogous Lemma 3.8 in \cite{chi2022faster} applies verbatim.
\end{proof}

\begin{lemma}\label{lemma:filter}
For every $l= h - 1, \dots, 0$ and $i, j \in [n]^2$, if we set $C_{i, j}^{(l)}$ to the minimum value of $A_{i, k}^{(l)} + B_{k, j}^{(l)}$ such that $k$ satisfies that:
\begin{itemize}
    \item $A^{(l + 1)}_{i, k} + B^{(l + 1)}_{k, j} = C^{(l+1)}_{i, j} + b$ for some $b \in \{-10, \dots, 10\}$
    \item $A_{i, k}^{*} + B^*_{k, j} = C^*_{i, j}$
\end{itemize}
then $C^{(l)}$ satisfies both properties (1) and (2).
\end{lemma}

\begin{proof}
The proof of correctness of Section 3.2.1 in \cite{chi2022faster} applies verbatim.
\end{proof}

\section{Applications}
\label{sec:applications}
In this section we list problems for which we can improve the bounds on the running time using \cref{thm:fast_rect}. We also improve running time for Unweighted Tree Edit Distance using the algorithm by \citet{chi2022faster} for Square Monotone Min-Plus Product. A summary of the improved bounds is given in \cref{fig:table_applications} but we detail the computations here.

\subsection{Single Source Replacement Paths}
\label{sec:SSRP}

In the SSRP problem, one is given an edge-weighted graph $G=(V, E, w)$ and a source $s \in V$ and has to find for all pairs $(v, e) \in V \times E$ the shortest path from $s$ to $v$ without using the edge $e$.
In \cite{polak}, \citeauthor{polak} use a Monotone Min-Plus Product algorithm to solve the \textbf{$M$-bounded SSRP} problem, which is the SSRP setting with edge-weights in $\{-M, \dots, M\}$
\footnote{As for Min-Plus Product, SSRP was shown in \cite{williams2010subcubic_focs, williams2010subcubic} to be sub-cubically equivalent to APSP. So under the APSP Hypothesis there is no truly sub-cubic time algorithm to solve SSRP in the general case.
\citet{polak} consider bounded weights to achieve a truly sub-cubic algorithm.}.

Their algorithm has a running time of (see \cite[Proof of Lemma 27]{polak}) $$\Tilde{O}(n^{\mu + \omega} + n^{\mu + \zeta + \omega(1-\zeta)} + n^{3 - 2\zeta} + n^{\varphi(1-\zeta, 1 + \mu)}),$$ where $M = n^\mu$ is the bound on the edge-weights, $\zeta\in[0, 1]$ is a parameter to optimize and $\varphi(\beta, \mu)$ is the function returning the smallest number such that the Monotone Min-Plus Product between matrices $A$ and $B$ of dimensions $n \times n^\beta$ and $n^\beta \times n$
with values of $B$ non-negative and at most $O(n^\mu)$ can be computed in $\Tilde O(n^{\varphi(\beta, \mu)})$ time. We recall that the interesting regime is when $0 \leq \mu \leq 3 - \omega$.
Using the bound on $\varphi(\beta, \mu)$ given in \cite{polak}, the running time is expressed as $\Tilde{O}(M^{\frac{5}{17-4\omega}} n^{\frac{36 - 7\omega}{17 - 4\omega}})$ using $\omega$ or as $\Tilde{O}(M^{0.8043}n^{2.4957})$ using known bounds for fast rectangular matrix multiplication.

\cref{thm:fast_rect} yields the tighter bound $\varphi(1-\zeta, 1 + \mu) \leq (3 + \mu - \zeta + \omega(1 - \zeta)) /2$. We can thus upper bound the running time of the algorithm by $\Tilde{O}(n^{\mu + \omega} + n^{\mu + \zeta + \omega(1-\zeta)} + n^{3 - 2\zeta} + n^{\frac{3 + \mu - \zeta + \omega(1 - \zeta)}{2}})$. Let's call $
A := \mu + \omega,\,
B := \mu + \zeta + \omega(1-\zeta),\,
C := 3 - 2\zeta \text{ and}\,
D := (3 + \mu - \zeta + \omega(1-\zeta))/2
$
the exponents of the different terms. We optimize $\zeta$ such that the exponent of $n$ equals $\min_{\zeta \in [0, 1]} \max \{A, B, C, D\}$ and will obtain running times of $\Tilde{O}(M^{\frac{2}{5 - \omega}} n^{\frac{9 - \omega}{5 - \omega}})$ using $\omega$ and $\Tilde{O}(M^{0.8825} n^{2.4466})$ using fast rectangular matrix multiplication. Remark that since $\mu \leq 3 - \omega$ the term $\Tilde O(n^A) = \Tilde O(Mn^\omega)$ is not dominant and thus can be omitted.\\

To express the running time as a function $\omega$, we bound $\omega(1 - \zeta) \leq (1 - \zeta) \cdot \omega + 2\zeta$ by implementing rectangular matrix multiplication by cutting rectangular matrices into square matrices and get:
\begin{align*}
    B &\leq \mu +\zeta  +  (1 - \zeta) \cdot \omega + 2\zeta \\
    &= \mu + \omega + \zeta \cdot (3 - \omega) &=: B' \\
    D &\leq (3 + \mu - \zeta + (1 - \zeta)\cdot \omega + 2\zeta)/2 \\
    &= (3 + \mu + \omega + \zeta \cdot (1 - \omega))/2 &=: D'.
\end{align*}
Now equalizing $B', C$ and $D'$ we find $\zeta = \frac{3 - \mu - \omega}{5 - \omega}$ which is indeed in $[0, 1]$ for $0 \leq \mu \leq 3 - \omega$. Hence the dominant term is $B' = \frac{2\mu + 9-\omega}{5 -\omega}$. This yields an expected running time of $\Tilde{O}(M^{\frac{2}{5 - \omega}} n^{\frac{9 - \omega}{5 - \omega}})$. \\

To improve this bound we can use fast rectangular matrix multiplication and known bounds on $\omega(\beta)$.
Indeed, since the function $\omega(\beta)$ is convex in $\beta$, we can compute an upper bound on $\omega(\beta_0)$ and $\omega(\beta_1)$ for $\beta_0 < \beta_1$ and interpolate an upper bound on $\omega(\beta) \leq \frac{\omega(\beta_0)(\beta_1 - \beta) + \omega(\beta_1)(\beta - \beta_0)}{\beta_1 - \beta_0}$ for values $\beta \in [\beta_0, \beta_1]$.
Let's set $\zeta$ as a linear function of $\mu$, i.e.\ $\zeta = a + b \cdot \mu$ for some real numbers $a, b \neq 0$.
Then by optimizing $\max \{B, C, D\}$ with plugged in $\mu_0=0$ we obtain the optimal value for $\zeta_0 = a$.
When plugging in $\mu_1 = 3 - \omega$ we get the optimal value for $\zeta_1$ and compute $b = \frac{\zeta_1 - \zeta_0}{3 - \omega}$. Using the current best bound on fast matrix multiplication \cite{gall_rect} we get $a = 0.2767$ and $b=-0.4412$.
Set $\beta = 1 - \zeta$ and choose the bounds $\beta_0 = 1 - \zeta_0 = 1 - a = 0.7233$ and $\beta_1 = 1$. We get the following upper bound for $\zeta \in [0, 0.2767]$.
\begin{align*}
    \omega(1 - \zeta) &\leq \frac{\omega(\beta_0)(\beta_1 - (1 - \zeta)) + \omega(\beta_1)(1 - \zeta - \beta_0)}{\beta_1 - \beta_0} \\
    &= \frac{\omega(\beta_0)(a + b \cdot \mu) + \omega(\beta_1)(-b \cdot \mu)}{a} \\
    &= \omega(\beta_0) + \frac{b}{a} (\omega(\beta_0) - \omega(\beta_1)) \mu
\end{align*}
Using the values of $a$ and $b$ previously computed and the bounds on fast matrix multiplication \cite{gall_square} and fast rectangular matrix multiplication \cite{gall_rect}, we obtain
$\omega(1-\zeta) \leq 2.1698 + 0.3237 \mu.$
Finally, if we use this bound to optimize $\max \{B, C, D\}$,
we get a running time of $\Tilde{O}(M^{0.8825}n^{2.4466})$.

\paragraph{Lower bound}
In \cite{polak}, \citeauthor{polak} have an interesting discussion about hardness of SSRP with negative weights as opposed to SSRP with positive weights. \citet{grandoni_william_ssrp} showed how to solve SSRP with positive weights in $\Tilde{O}(Mn^\omega)$ time, where $M$ is a bound on the weights. Subsequently this means that for weights in $\{0, 1\}$, we can solve SSRP in $\Tilde{O}(n^\omega)$ time. \citeauthor{polak} show that SSRP with weights in $\{-1, 0, 1\}$ seems harder than SSRP with weights in $\{0, 1\}$. They prove that if there exists a $T(n)$ time algorithm for SSRP in an $n$-vertex graphs with edge-weights in $\{-1, 0, 1\}$, then there exists an $O(T(n)\sqrt{n})$ time algorithm for the Bounded-Difference Min-Plus Product of $n \times n$ matrices. So if SSRP with $\{-1, 0, 1\}$ weights can be solved in $\Tilde{O}(n^2)$ time (if $\omega = 2$), then the Bounded-Difference Min-Plus Product could be solved in $\Tilde{O}(n^{2.5})$. The existence of the second algorithm seemed unlikely at the time of publication of \cite{polak}. However the new algorithm by \citet{chi2022faster} yields an $\Tilde {O}(n^{2.5})$ time algorithm if $\omega = 2$. Hence the above reasoning to lower bound the running time of SSRP with negative weights by $\Tilde{O}(n^2)$ (if $\omega = 2$) does not hold anymore.

\subsection{Batch Range Mode}

Let $a$ be an array of elements. A range mode query asks for the most frequent element in a given contiguous interval of $a$. In the \textbf{Batch Range Mode} problem, the array $a$ and all queries are given in advance, and the task is to compute all query answers.

The Batch Range Mode problem is solved by \citet{batch_range_mode} via a black box computation of a Monotone Min-Plus Product in time $\Tilde{O}(n^{\frac{27+2\omega}{19+\omega}})$. 
%
\citet{polak} bound this running time by plugging in their algorithm for Monotone Min-Plus Product and obtained a running time $\Tilde O(n^\frac{21 + 2\omega}{15 + \omega})$. \citet{Gao_He_batch} showed a faster algorithm running in time $\Tilde O(n^\frac{18+2\omega}{13 + \omega})$. 
%
To improve the bound even further, we add a small modification to the algorithm of \citet{batch_range_mode} and use the (square) Monotone Min-Plus Product algorithm of \citet{chi2022faster}.

We first sketch the algorithm of \citet{batch_range_mode}. For some parameter $T \leq n$, denote an element of $a$ as \emph{infrequent} if it appears at most $T$ times, and \emph{frequent} if it appears more than $T$ times. 
A persistent binary search tree is used to handle range queries whose answer is an infrequent element. This takes time $\Tilde O(nT)$ for all queries. 
To handle range queries whose answer is a frequent element, split the array $a$ into $O(n/T)$ consecutive blocks of size $O(T)$.
For any element of $a$ in a given range, there are now two cases: either the element lies in a block fully contained in the range (case 1), or it lies in a block that is not fully contained in the range (case 2). Since blocks have length $O(T)$, there are at most $O(T)$ elements in case 2. So by using binary search trees we can answer range queries whose answer is in case 2 in time $\Tilde O(T)$ per query, i.e.~time $\Tilde O(nT)$ overall. 
To answer all range queries whose answer is in case 1, it suffices to compute the Max-Plus Product between two matrices $A$ and $B$, whose construction is detailed in \citet{batch_range_mode}. 
Importantly for the current analysis, both $A$ and $B$ are $O(n/T)$ by $O(n/T)$ matrices and $B$ is row-monotone. Additionally, the values of $A$ and $B$ is the number of occurrences of frequent elements over certain ranges, so the entries of $A$ and $B$ are at most $n$. 

In total, the algorithm of \citet{batch_range_mode} runs in time $\Tilde O(n^{1 + \tau} + n^{(1-\tau) \varphi(1, \frac{1}{1-\tau})})$, where $\tau \in (0, 1)$ is such that $T = n^\tau$ and $\varphi(\beta, \mu)$ is the function returning the smallest number such that the Monotone Min-Plus Product between matrices $A$ and $B$ of dimensions $n\times n^\beta$ and $n^\beta \times n$ with values of $B$ non-negative and at most $O(n^\mu)$ can be computed in $\Tilde O(n^{\varphi(\beta, \mu)})$ time.
Using the bound of \cref{thm:fast_rect}, this yields running time $\Tilde O(n^{\frac{5+2\omega}{4+\omega}})$, which is worse than the bound achieved by \citet{Gao_He_batch}. However, consider the following modification. 

Differentiate between \emph{infrequent}, \emph{frequent} and \emph{highly frequent} elements as elements of $a$ that appear at most $T$ times, more than $T$ but at most $n/T$ times, and more than $n/T$ times respectively. Then the infrequent and frequent elements can be handled in the same way as in \citet{batch_range_mode}. Since frequent elements now appear at most $n/T$ times, the entries of the constructed matrices $A$ and $B$ is at most $n/T$. 
Finally, there are at most $T$ highly frequent items. So, similarly to frequent elements in case 2, we can use binary search trees to handle highly frequent elements in time $\Tilde O(T)$ per query, i.e.~$\Tilde O(nT)$ time overall. 

In total, this yields a running time $\Tilde O(n^{1+\tau} + n^{(1-\tau) \varphi(1, 1)})$. Using the bound $\varphi(1, 1) \leq \frac{3+\omega}{2}$ of \citet{chi2022faster}, and setting $\tau = \frac{1+\omega}{5+\omega}$, the running time becomes $\Tilde O(n^{\frac{6+2\omega}{5+\omega}})$. 
With current best bound on fast matrix multiplication \cite{gall_square} this is $\Tilde{O}(n^{1.4575})$.

\subsection{$k$-Dyck Edit Distance}

The Dyck Edit Distance of a sequence $S$ of parenthesis (of various types) is the smallest number of edit operations (insertions, deletions, and substitutions) needed to transform $S$ into a sequence of balanced opening and closing parenthesis. \citet{dyck_arxiv} study the \textbf{$k$-Dyck Edit Distance} problem in which the input is an $n$-length sequence of parenthesis $S$ and a positive integer $k$ and the task is to compute the Dyck Edit Distance of $S$ if it is at most $k$, otherwise to return $k+1$.

\citeauthor{dyck_arxiv} show how to solve it in time $\Tilde O(n + k^2 \cdot {\psi(k, \sqrt{k})})$ (see \cite[Theorem 1.1]{dyck_arxiv}), where $\psi(a, b)$ is the running time to compute the Min-Plus Product between two bounded-difference matrices of dimensions $a \times b$ and $b \times a$.
To bound $\psi(a, b)$, \citeauthor{dyck_arxiv} use the algorithm for Bounded-Difference Min-Plus Product in \citet{rna_led_osg} and get a running time of $\Tilde O(n+k^{4.7820})$.

Our new result of \cref{thm:fast_rect} yields a tighter bound of $\psi(n, n^\beta) \leq \Tilde O(n^\frac{2 + \beta + \omega(\beta)}{2})$. Hence the running time improves to $\Tilde O(n + k^{\frac{5}{2} + \omega(1/2)})$, which is $\Tilde O(n + k^{4.5442})$ with the current best bound on fast rectangular matrix multiplication \cite{gall_rect}. Actually if $k \geq \sqrt{n}$, \citeauthor{dyck_arxiv} show a faster algorithm running in time $\Tilde O(\psi(n, k))$ (see \cite[Section 4]{dyck_arxiv}), which we can directly bound to $\Tilde O(n^{\frac{2 + \log_n k + \omega(\log_n k)}{2}}) = \Tilde O(n \sqrt{k n^{\omega(\log_n k)}})$.

The same improved bound was obtained independently and in parallel by \citet{dyck_soda}.

\subsection{2-approximation APSP}

The All Pairs Shortest Path (APSP) problem asks for a given graph $G = (V, E)$ the shortest path between all pairs of vertices. \citet{apsp_approx} consider a \textbf{2-approximation to APSP} for undirected unweighted graphs, which is a mapping $\hat d: V \times V \rightarrow \mathbb{N}$  such that $d_{OPT}(u, v) \leq \hat d(u, v) + 2$ for every vertices $u$ and $v$, where $d_{OPT}(u, v)$ is the shortest path from $u$ to $v$.

The 2-approximation to APSP is computed by \citeauthor{apsp_approx} by using Bounded-Difference Min-Plus Product as a black-box. When $n = |V|$, the running time of their algorithm is $\Tilde O(n^2 t^{1/2} + \psi(n, n/t))$ (see \cite[Theorem 11]{apsp_approx}), where $t \geq 1$ is a parameter and $\psi(a, b)$ is the running time to compute the Min-Plus Product between two bounded-difference matrices of dimensions $a \times b$ and $b \times a$. \citeauthor{apsp_approx} use the algorithm for Bounded-Difference Min-Plus Product in \citet{rna_led_osg} and get a running time of $\Tilde O(n^{2.2867})$. With \cref{thm:fast_rect} we have an improved upper bound on $\psi$ and can thus compute a 2-approximation of APSP in $\Tilde O(n^2 t^{1/2} + n^{\frac{2 + \log_n(n/t) + \omega(\log_n(n/t))}{2}})$ time. By setting $t = n^{0.5185}$, the running time becomes $\Tilde O(n^{2.2593})$ using current best bounds on fast rectangular matrix multiplication \cite{gall_rect}.

\subsection{Unweighted Tree Edit Distance}

Given an alphabet $\Sigma$ and a rooted ordered tree with node labels in $\Sigma$, we consider two types of operations: relabeling a node to another symbol of $\Sigma$ and deleting a node $v$ such that the children of $v$ become children of the parent of $v$ with the same ordering.
The \textbf{Unweighted Tree Edit Distance} between two rooted ordered trees $T_1$ and $T_2$ is the minimum number of operations to perform on $T_1$ and $T_2$ so that they become identical
\footnote{In the weighted case, the cost of deletions and relabeling is a function of the symbol on the nodes. \citet{weighted_ted_soda, weighted_ted} show that the Weighted Tree Edit Distance does not admit any truly subcubic algorithm under the APSP Hypothesis, if $T_1$ and $T_2$ have both $n$ nodes and $|\Sigma|=\Theta(n)$.
}.

In \cite{treeED}, \citeauthor{treeED} shows the first truly sub-cubic algorithm to solve the Unweighted Tree Edit Distance. For two trees of size $n$ and $m$, the simplified expression of the running time is $\Tilde O(n \cdot m^\frac{3\alpha -1}{\alpha +1})$ (see \ref{app:ted}), where $\alpha$ is the exponent of the Bounded-Difference Min-Plus Product between square matrices. Using the algorithm for Bounded-Difference Min-Plus Product of \citet{rna_led_osg}, the running time can be bounded by $\Tilde O(n \cdot m^{1.9541})$ time, which is subcubic if $m=O(n)$. We remark here that the running time $\Tilde O(n \cdot m^{1.9546})$ announced in \cite{treeED} is slightly worse due to a numerical error.

Since a $n \times n$ bounded-difference matrix can be transformed into a row-monotone matrix with non-negative entries bounded by $O(n)$, the work of \citet{chi2022faster} allows us to get a tighter bound on $\alpha \leq \frac{3 + \omega}{2}$. We can thus bound the running time by $\Tilde O(n \cdot m^\frac{3\omega + 7}{\omega + 5})$. With the current best bound on fast matrix multiplication \cite{gall_square}, this is bounded by $\Tilde O(n \cdot m^{1.9149})$.

\section*{Acknowledgment}
I would like to thank Adam Polak for introducing the topic to me and for instructive guidance. I'm also thankful to Ovidiu Rața for pointing out an error in the first version of this paper concerning the Batch Range Mode problem.

\bibliography{ref.bib}

\appendix

\section{Simplified Running Time for Unweighted Tree Edit Distance}
\label{app:ted}

We explain here how the running time of the algorithm for Unweighted Tree Edit Distance of \citet{treeED} can be expressed as a function of the running time for Bounded-Difference Min-Plus Product.

We use the notations of the article \cite{treeED_arxiv} without redefining them. In \cite[Section 4.2.3]{treeED_arxiv}, the running time of \cite[Algorithm 2]{treeED_arxiv} is shown to be $\Tilde O(\frac{n}{\Delta}\text{MUL}(\Delta, m) + n \cdot m \cdot \Delta^2)$, with $\Delta$ a parameter to optimize, $n = |T_1|$, $m = |T_2|$ and $\text{MUL}(m, n)$ is the function defined in \cite[Theorem 4.5]{treeED_arxiv}. This theorem gives a (numerical) bound on $\text{MUL}(m, n)$ which is proven in \cite[Section 4.4, Algorithm 6]{treeED_arxiv}. The running time of \cite[Algorithm 6]{treeED_arxiv} is $\Tilde O(n^2 \cdot m^2 / \delta)$ for the recursive calls and $\Tilde O(n^2 \cdot \delta^{\alpha - 2})$ for the final part, where $\alpha$ is the exponent for Bounded-Difference Min-Plus Product and $\delta$ is a parameter to optimize
\footnote{In \cite[Section 4.4]{treeED_arxiv} the parameter is called $\Delta$, but we rename it to $\delta$ to avoid confusion with the $\Delta$ parameter of the running time of \cite[Algorithm 2]{treeED_arxiv}.}.
Here, we do not replace $\alpha$ by the bound given by \citet{rna_led_osg} as done in \cite{treeED_arxiv}, but directly bound $\text{MUL}(m, n)$ by $\Tilde O(n^2 \cdot m^2 / \delta + n^2 \delta^{\alpha - 2})$. To equalize both terms we set $\delta = m^\frac{2}{\alpha - 1}$ and get the bound $\text{MUL}(m, n) = \Tilde O(n^2 \cdot m^\frac{2\alpha -4}{\alpha -1})$.
We can now replace $\text{MUL}(\Delta, m)$ in the total running time of \cite[Algorithm 2]{treeED_arxiv} and obtain $\Tilde O(n \cdot m^2 \cdot \Delta^\frac{\alpha - 3}{\alpha -1} + n \cdot m \cdot \Delta^2)$. We equalize both terms by setting $\Delta = m^\frac{\alpha -1}{\alpha +1}$ and get the following expression of the running time for Unweighted Tree Edit Distance: $\Tilde O(n \cdot m^\frac{3\alpha -1}{\alpha +1})$.

One can verify that we obtain the same result as \cite[Theorem 4.5]{treeED_arxiv} and \cite[Theorem 1.1]{treeED_arxiv} if we replace $\alpha$ in the final expression by the numerical bound given by  \citet{rna_led_osg}: $\alpha \leq 2.8244$ in the randomized algorithm and $\alpha \leq 2.8603$ for the deterministic case\footnote{This is up to a small numerical error done by \cite{treeED_arxiv}. In fact when applying \cite[Theorem 4.5]{treeED_arxiv} in \cite[Section 4.2.3]{treeED_arxiv}, the exponent of $\Delta$ should be 0.0962 and not 0.0952.}.

\section{Min-Plus Product when B is column-monotone}
\label{sec:BMMP_col}
We want to compute the Min-Plus Product $C = A \star B$ of the  $n^\alpha \times n$ matrix $A$ and the $n \times n^\gamma$ matrix $B$ where the columns of $B$ are monotonously non-decreasing and entries of $B$ are positive and bounded by $O(n^\mu)$. Without detailing every computation again, we explain why $C$ can be computed in $\Tilde O(n^{\frac{\omega(\alpha, 1, \gamma) + \alpha + \gamma + \mu}{2}})$ time, where $\omega(a, b, c)$ is the lowest exponent of the multiplication of matrices of dimensions $n^a \times n^b$ and $n^b \times n^c$.
If matrices $A$ and $B$ have dimension $n \times n^\beta$ and $n^\beta \times n$ respectively, the running time becomes $\Tilde O(n^{\frac{\omega(\beta) + 2 + \mu}{2}})$. Note that this is different from the $\Tilde O(n^{\frac{\omega(\beta) + 1 + \beta + \mu}{2}})$ time algorithm for the Min-Plus Product when $B$ is row-monotone. Intuitively, the algorithm iterates over the rows of $A$, over the rows of $B$ (which corresponds to the columns of $A$) and the columns of $B$. When $B$ is row-monotone, iterating over the columns of $B$ can be reduced to a number of steps depending on $\mu$ and iterating over the rows of $A$ and the rows of $B$ takes $O(n \cdot n^\beta)$ time. However when $B$ is column-monotone, the acceleration happens when iterating over the rows of $B$ and hence iterating over the two other ranges takes $O(n \cdot n)$ time.

\paragraph{Quotient matrix}
The algorithm starts by sampling a prime number $p$ from the range $[40n^\beta, 80n^\beta]$ uniformly at random, where $\beta \in (0, \mu)$ is some parameter we tune later. Then we can still assume \cref{ass:p/3}. We first compute the quotient matrix $\Tilde C = \Tilde A \star \Tilde B$ where $\Tilde A_{ij} = \left\lfloor \frac{A_{ij}}{p} \right\rfloor$ if $A_{ij}$ is finite, otherwise $\Tilde A_{ij} = + \infty$, and similarly for $\Tilde B$. The entries of $\Tilde B$ are bounded by $O(n^{\mu - \beta})$ so since the number of intervals of $+\infty$ values in each column of $\Tilde B$ is bounded by $O(n^{\mu - \beta})$, there are $O(n^{\mu - \beta})$ intervals of the same value in each column of $\Tilde B$. We can thus compute $\Tilde C$ using a segment tree structure by iterating over all $i \in [n^\alpha]$, over all $j \in [n^\gamma]$ and over all intervals $[k_0, k_1]$ of the $j$th column of $\Tilde B$. This computation takes $O(n^{\alpha + \gamma + \mu - \beta})$ time. By \cref{ass:p/3}, $\Tilde C_{ij} = \left\lfloor \frac{C_{ij}}{p} \right\rfloor$ if $C_{ij}$ is finite, otherwise $\Tilde C_{ij} = +\infty$.

\paragraph{Remainder matrix}
We then compute the remainder matrix $(C \mod p)$ recursively. Let $h$ be the integer such that $2^{h - 1} \leq p < 2^h$ and define for $\ell = 0, 1, \dots, h$ matrices $A^{(\ell)}$ and $B^{(\ell)}$ where $A_{ij}^\ell = \left\lfloor \frac{A_{ij} \mod p}{2^\ell} \right\rfloor$ if $A_{ij}$ is finite, otherwise $A_{ij}^\ell = +\infty$, and similarly for $B^{(\ell)}$. Observe that the values of $A^{(\ell)}$ and $B^{(\ell)}$ are bounded by $O(n^\beta)$ and that each column of $B^{(\ell)}$ contains at most $O(n^\mu / 2^\ell)$ intervals of the same value.
We define a segment w.r.t.\ $\ell$  to be a tuple $(i, [k_0, k_1], j)$ such that for any $k \in [k_0, k_1]$ $A_{ik}^{(\ell)} = A_{ik_0}^{(\ell)}$, $\Tilde A_{ik} = \Tilde A_{ik_0}$, $B^{(\ell)}_{kj}= B^{(\ell)}_{k_0j}$ and $\Tilde B_{kj} = \Tilde B_{k_0}j$. For a fixed $i$ and fixed $j$, there are $O(n^\mu / 2^\ell)$ segments. Next define the set $T_b^{(\ell)}$ to be the set of segments w.r.t.\ $\ell$ such that $A^{(\ell)}_{ik_0} + B^{(\ell)}_{k_0j} = C^{(\ell)}_{ij} + b$ but $\Tilde A_{ik_0} + \Tilde B_{k_0j} \neq \Tilde C_{ij}$ for some offset $b \in [-10, 10]$.
We recursively compute the matrix $C^{(\ell)}$ where $C^{(\ell)}_{ij} = \left\lfloor \frac{(C_{ij} \mod p)\  \pm \  2(2^\ell - 1)}{2^\ell} \right\rfloor$ for decreasing values of $\ell = h, h-1, \dots, 0$.

\subparagraph{Base case}
In the base case $\ell = h$, since $p < 2^h$ all matrices $A^{(h)}$, $B^{(h)}$ and $C^{(h)}$ have zero entries. $T_0^{(h)}$ is thus the set of all segments w.r.t.\ $h$ $(i, [k_0, k_1], j)$ such that $\Tilde A_{ik_0} + \Tilde B_{k_0j} \neq \Tilde C_{ij}$ and $T_b^{(h)}$ is the empty set for $b\neq 0$ . $T_0^{(h)}$ can be computed in $O(n^{\alpha + \gamma + \mu - \beta})$ time by iterating over all the $i \in [n^\alpha]$, all the $j \in [n^\gamma]$ and all the intervals of the $j$th column of $\Tilde B$.

\subparagraph{Recursive case}
Each recursive step consists of computing  $C^{(\ell)}$ from $T_b^{(\ell + 1)}$ and then computing $T_b^{(\ell)}$ from $T_b^{(\ell + 1)}$.
To compute $C^{(\ell)}$ from $T_b^{(\ell + 1)}$ we first use fast matrix multiplication on matrices of polynomials, then we subtract the erroneous terms and finally look for the lowest value. The matrices we multiply are $A^p$ and $B^p$ where $A_{ik}^p = x^{A_{ik}^{(\ell)} - 2A_{ik}^{(\ell + 1)}} y^{A_{ik}^{(\ell)}}$ if $A_{ik}$ is finite and $0$ otherwise, and similarly for $B^p$. The $x$-degree is 0 or 1 and the $y$-degree is bounded by $O(n^\beta)$. Computing $C^p = A^p \times B^p$ thus takes $\Tilde O(n^{\omega(\alpha,1,\gamma) + \beta})$ time. We then compute a polynomial of candidate terms $C_{ijb}^p = \sum_{\substack{\lambda x^c y^d \in C^p \\ d = C_{ij}^{(\ell + 1)} + b}} \lambda x^c$ for every $i \in [n^\alpha]$, every $j \in [n^\gamma]$ and every $b \in [-10, 10]$ . Since $C_{ij}^{(\ell + 1)}$ can take up to $O(n^\beta)$ values, this takes $O(n^{\alpha + \beta + \gamma})$ time. Finally we use $T_b^{(\ell + 1)}$ to compute the set of erroneous terms $R_{ijb}^p = \sum_{(i, [k_0, k_1], j) \in T_b^{(\ell + 1)}} x^{A_{ik}^{(\ell)} + B_{kj}^{(\ell)} - 2(A_{ik}^{(\ell + 1)} + B_{kj}^{(\ell + 1)})}$. This takes time $O(|T_b^{(\ell + 1)}|)$.
It still holds that $\bigcup\limits_{b = -10}^{10} T_b^{(\ell)} = \bigcup\limits_{b=-10}^{10} T_b^{(\ell + 1)}$ and since each segment in $T_b^{(\ell + 1)}$ breaks into a constant number of segments of $T_{b'}^{(\ell)}$ (for $b' \in [-10, 10]$), the algorithm simply needs to sort the sub-segments in the correct set, which can be done using binary search. Therefore computing all $T_b^{(\ell)}$ from $T_b^{(\ell + 1)}$ can be done in $O(|T_b^{(\ell + 1)}|)$.

\paragraph{Total runtime}
Finally, we need to bound the size of $T_b^{(\ell)}$. We can use the same argument as previously using a prime number theorem. A fixed segment such that $\Tilde A_{ik_0} + \Tilde B_{k_0j} \neq \Tilde C_{ij}$ is in $T_b^{(\ell)}$ with probability $\Tilde O(2^\ell/n^\beta)$ and there are $O(n^{\alpha + \gamma + \mu}/ 2^\ell)$ such segments so we can bound the expected size of $T_b^{(\ell)}$ by $\mathbb{E}_p (|T_b^{(\ell)}|) = \Tilde O(n^{\alpha + \gamma +\mu -\beta})$.
In the end the total running time is $\Tilde O(n^{\omega(\alpha, 1, \gamma) + \beta} + n^{\alpha + \gamma + \mu - \beta} + n^{\alpha + \gamma + \beta})$. Optimizing over $\beta$ this yields $\Tilde O(n^{\frac{\omega(\alpha, 1, \gamma) + \alpha + \gamma + \mu}{2}})$.

\end{document}